\documentclass[11pt, a4paper]{article}
\usepackage[utf8]{inputenc}
\usepackage[usenames,dvipsnames]{xcolor}
\usepackage{amsmath}
\usepackage{amsthm}
\usepackage{amsfonts}
\usepackage{amssymb}
\usepackage{array}  
\usepackage{color}
\usepackage[english]{babel}
\usepackage{mathrsfs}
\usepackage{graphicx}
\usepackage{dsfont}
\usepackage{geometry}
\definecolor{orangebis}{rgb}{0.99,0.25,0.00}
\definecolor{greenbis}{rgb}{0.10,0.85,0.10}
\definecolor{bluebis}{rgb}{0.10,0.30,0.99}
\usepackage[final]{hyperref}
\usepackage{subcaption}
\usepackage{appendix}
\usepackage{chngcntr}
\usepackage{etoolbox}

\theoremstyle{plain}
\newtheorem{theorem}{Theorem}
\counterwithin{theorem}{section}

\newtheorem{lemma}[theorem]{Lemma}

\theoremstyle{definition}

\newtheorem{remark}[theorem]{Remark}

\newcommand{\R}{\mathbb{R}}

\newcommand{\Z}{\mathbb{Z}}

\newcommand{\E}{\mathbb{E}}

\def\calE{\mathcal{E}}

\newcommand{\vol}{\textup{Vol}}
\newcommand{\prob}{\mathbb{P}}

\newcommand{\un}{\mathds{1}}

\renewcommand{\textbf}[1]{\begingroup\bfseries\mathversion{bold}#1\endgroup}
\def\var{\mathop{\mathrm{Var}}}

\title{A lower bound for the Bogomolny and Schmit constant for random monochromatic plane waves}
\date{\today}
\author{Maxime Ingremeau and Alejandro Rivera}

\begin{document}

\maketitle

\section{Introduction}
Let $f: \R^2 \rightarrow \R$ be the stationary, isotropic planar centered almost surely continuous Gaussian field\footnote{We refer the reader to \cite{abrahamsen1997review} for the definition and properties of random fields.} with covariance $\E[f(x)f(y)]=J_0(|x-y|)$ where $J_0(r)=\frac{1}{2\pi}\int_0^{2\pi}e^{ir\sin(\theta)}d\theta$ is the $0$th Bessel function of the first kind. We call this field \textbf{the random monochromatic plane wave}. For each $R>0$, let us denote by $B(0,R)$ the ball of centre $0$ and of radius $R$, and by $N(R,f)$ the number of connected components of $\R^2\backslash f^{-1}(\{0\})$ included in $B(0,R)$.

In \cite{ns_2009}, \cite{ns_2016}, Nazarov and Sodin showed that the limit
\begin{equation}
\nu_{BS} = 4\pi\times\lim_{R\rightarrow \infty} \frac{\E N(R,f)}{\pi R^2}
\end{equation}
exists, and is positive. However, their method does not give an explicit value for the constant $\nu_{BS}$, sometimes called the \emph{Bogomolny-Schmit constant} (or also Nazarov-Sodin constant). An equivalent definition for $\nu_{BS}$ is that it is the limit of the average number of nodal domains for a random spherical harmonic divided by the degree of this spherical harmonic. The factor $4\pi$ can be interpreted as the area of the unit sphere. Bogomolny and Schmit gave in \cite{bs_2002} a highly heuristical argument, based on percolation, which yielded the value
\begin{equation*}
\nu_{BS} \simeq 0.0624.
\end{equation*}

However, numerical simulations carried out by Nastasescu (\cite{nastasescu_thesis}), Konrad (\cite{konrad_thesis}) and Beliaev and Kereta (\cite{bk_2013}) showed that $\nu_{BS}\simeq 0.0589$, contradicting Bogomolny and Schmit's prediction by a few percents. Experimental measurements of $\nu_{BS}$ were also realised by Kuhl, H{\"o}hmann, St{\"o}ckmann and Gnutzmann (\cite{kuhl2007nodal}), using Berry's conjecture (\cite{berry1977regular}) that high frequency eigenmodes in a chaotic cavity should locally behave like the monochromatic random wave $f$. From the mathematical point of view, though, very little is known regarding the value of $\nu_{BS}$. The best rigorous lower bound so far was the one given in \cite{nastasescu_thesis}, which is of the order of $10^{-319}$. The aim of this note is obtain a much better lower bound by elementary means:
\begin{theorem}\label{t.main}
\begin{equation*}
\nu_{BS}\geq 1.39\times 10^{-4}.
\end{equation*}
\end{theorem}

This bound is much smaller than the expected value of $\nu_{BS}$. 
However, our method does not take into account all nodal domains, but only those which are included in circles of radius $3.8$ (the first minimum of the Bessel function $J_0$). After visual inspection of computer simulations we expect that these are not very common. Aside from this, in our use of Lemma \ref{lem:BoundDist}, we ignore the fact that small, isolated nodal domains should be somewhat rare. We hope that our methods can be used to count more general nodal domains, and to obtain sharper lower bounds on $\nu_{BS}$.\\

Before proving Theorem \ref{t.main} let us recall two basic facts about random monochromatic plane waves.
\begin{itemize}
\item The function $f$ almost surely satisfies $\Delta f+ f=0$ where $\Delta$ is the Laplace operator. This follows by applying the Laplace operator with respect to $x$ to the expression $\E[f(x)f(y)]=J_0(|x-y|)$ because $J_0$ is also an eigenfunction of the Laplace operator.
\item In polar coordinates, $f$ can be given the simple expression 
\begin{equation}\label{eq:GaussianField}
f(r, \theta)=X_0 J_0(r) - \sqrt{2} \sum_{n\geq 1}J_n(r)\left(X_n\cos(n\theta)+Y_n\sin(n\theta)\right)\, .
\end{equation}
where for each $n\in\Z$ we denote by $J_n$ the $n$th Bessel function of the first kind, that is $J_n(x)=\frac{1}{2\pi}\int_0^{2\pi} e^{in\theta -x\sin(\theta)}d\theta$, and where $(X_k)_{k\geq 0}$ and $(Y_k)_{k\geq 1}$ are two families of centred Gaussian random variables with unit variance and independent as a whole. An explanation for this fact can be found for instance in Section 4.2 of \cite{ns_10}.
\end{itemize}
\paragraph{Acknowledgement}
The authors would like to thank A. Deleporte for suggesting the proof of Lemma \ref{lem:BoundDist}, as well as M. Sodin, S. Muirhead and M. McAuley for pointing out mistakes in previous versions of this note.

The first author was funded by the Labex IRMIA, and partially supported by the Agence Nationale de la Recherche project GeRaSic
(ANR-13-BS01-0007-01).\\
The second author was funded by the ERC grant Liko No 676999.

\section{Proof of Theorem \ref{t.main}}
We now prove Theorem~\ref{t.main}. The idea is to consider points such that $f$ does not vanish on a certain circle around this point. We start by producing a lower bound for the probability that $0$ is such a point.
\begin{lemma}\label{l.good}
Fix $r>0$ and write $\Psi:t\mapsto\int_t^{+\infty}e^{-t^2/2}\frac{dt}{\sqrt{2\pi}}$. Let $\calE_r$ be the event that the zero set of $f$ does not intersect the circle of radius $r$. Then, for each $T\in\R$,
\[
\prob\left[\calE_r\right]\geq 2\Psi(T)- \sqrt{2}r\Psi\left(\frac{T}{\sqrt{1-J_0(r)^2}}\right)\, .
\]
\end{lemma}
\begin{proof}
In this proof we use the expression \eqref{eq:GaussianField} for $f$. Fix $r>0$ and for each $\theta\in[0,2\pi]$, let $u(\theta)=X_0J_0(r)-f((r\cos(\theta),r\sin(\theta))$. We first fix $x_0>0$ and try to estimate the probability that $u(\theta)$ crosses the level $x_0 J_0(r)$ when $\theta$ varies on the unit circle. Also, throughout our calculations, we will use the following two Bessel function identities:
\begin{equation}\label{e.identities}
J_{-n}(x)=(-1)^nJ_n(x);\, \sum_{n\in\Z} J_n(x)^2=1;\, \sum_{n\in\Z} n^2 J_n(x)^2 = \frac{x^2}{2}\, .
\end{equation}
These identities follow from the following classical formula (cf. \cite[Chapter 2]{watsonn})
\[
\sum_{n\in\Z} t^n J_n(x) = e^{\frac{x}{2}(t-t^{-1})}
\]
by setting $t=e^{i\theta}$ and applying Parseval's formula.\\
Now, observe that, since $f$ is stationary, isotropic and $(f,\nabla f)$ is non-degenerate, $u$ is a stationary Gaussian field on the unit circle and the field $(u,u')$ is non-degenerate so that $u$ crosses the level $x_0J_0(r)$ almost surely an even number of times. In particular, by Markov's inequality,
\[
\prob\left[\exists\theta;\, u(\theta)=x_0J_0(r)\right]\leq\frac{1}{2}\E\left[\textup{Card}\{\theta\in [0,2\pi];\, u(\theta)=x_0J_0(r)\}\right]\, .
\]
To compute the right-hand side of this inequality, we apply the Kac-Rice formula (see Theorem 6.2 from \cite{azws}). Define $\alpha(r)>0$ by $\alpha(r)^2 = \var(u(\theta))=2\sum_{n\geq 1} J_n(r)^2$. Then, by the first and second identity in \eqref{e.identities}, $\alpha(r)^2 =1-J_0(r)^2$. By the Kac-Rice formula,
\[
\E\left[\textup{Card}\{\theta\in [0,2\pi];\, u(\theta)=x_0J_0(r)\}\right]=\int_0^{2\pi}\E\left[|u'(\theta)|\, |\, u(\theta)=x_0J_0(r)\right]\frac{e^{-\frac{x_0^2 J_0(r)^2}{2\alpha(r)^2}}}{\alpha(r)\sqrt{2\pi}}d\theta\, .
\]
The fact that $u$ is stationary, implies first that the integrand is independent of $\theta$ and second, that $u'(\theta)$ is independent of $u(\theta)$. Observe that $\var\left(u'(\theta)\right)=2\sum_{n\geq 1} n^2 J_n(r)^2$. By the first and third identities in \eqref{e.identities}, $\var(u'(\theta))=\frac{r^2}{2}$. Moreover, if $\xi$ is a centred Gaussian of variance one, $\E\left[|\xi|\right]=\sqrt{\frac{2}{\pi}}$. Therefore,
\[
\E\left[\textup{Card}\{\theta\in [0,2\pi];\, u(\theta)=x_0J_0(r)\}\right]=\frac{\sqrt{2}r}{\alpha(r)}\exp\left(-\frac{x_0^2J_0(r)^2}{2\alpha(r)^2}\right)\, .
\]

On the other hand, we have
\begin{equation*}
\mathbb{P}[\neg\mathcal{E}_r;\ X_0=x_0] = \mathbb{P}\big{[} \exists \theta\in [0,2\pi];\, u(\theta)=x_0J_0(r)\big{]} \leq \E\left[\textup{Card}\{\theta\in [0,2\pi];\, u(\theta)=x_0J_0(r)\}\right]
\end{equation*}

In particular, for each $T\in\R$, the probability that $f$ has a zero on the circle of radius $r$ centred at $0$ satisfies
\begin{equation}\label{eq:LowerProp}
\prob\left[\calE_r\right]\geq \E\left[\left(1 - \frac{r}{\sqrt{2}\alpha(r)}\exp\left(-\frac{X_0^2J_0(r)^2}{2\alpha(r)^2}\right)\right)\un_{[|X_0|\geq T]}\right]\, .
\end{equation}
Let $\Psi(t)=\prob[X_0\geq t]=\int_t^{+\infty}e^{-t^2/2}\frac{t}{\sqrt{2\pi}}$. Then, the right-hand side is
\[
2\Psi(T)-2\frac{r}{\sqrt{2}\alpha(r)}\E\left[\exp\left(-\frac{X_0^2J_0(r)^2}{2\alpha(r)^2}\right)\un_{[X_0\geq T]}\right]\, .
\]
But for each $a>0$, a simple calculation shows that $\E\left[\exp(-aX_0^2)\un_{[X_0\geq T]}\right]=\frac{1}{\sqrt{1+2a}}\Psi\left(\sqrt{1+2a}T\right)$ so
\[
\prob\left[\calE_r\right]\geq 2\Psi(T)-\frac{\sqrt{2} r}{\sqrt{\alpha(r)^2+J_0(r)^2}}\Psi\left(\sqrt{1+J_0(r)^2/\alpha(r)^2}T\right)
\]
Replacing $\alpha(r)$ by its expression, we get, for any $T>0$ and any $r>0$,
\begin{equation}
\prob\left[\calE_r\right]\geq 2\Psi(T)-\sqrt{2} r\Psi\left(\frac{T}{\sqrt{1-J_0(r)^2}}\right)\, .
\end{equation}
\end{proof}
Let $G_r=G_r(f)\subset \R^2$ be the (random) set of points $x\in\R^2$ for which $f$ does not vanish on the circle of radius $r$ centred at $x$.  If $x, y\in \R^2$, we will write $x\sim_r y$ if $x,y\in G_r$ and $x,y$ belong to the same connected component of $\R^2\backslash f^{-1}(\{0\})$. The next step of the proof is to show that the connected components of $G_r$ are not too large.
\begin{lemma}\label{lem:BoundDist}
Let $r_1$ and $r_2$ denote the first and second zeros of $J_0$ and let $r\in]r_1,r_2[$. Then, for each $x\in G_r$, the connected component of $x$ in $\R^2\backslash f^{-1}(\{0\})$ is included in $B(x,r)$. In particular, if $x,y\in G_r$ are such that $x\sim_r y$, then $|x-y|\leq r$.
\end{lemma}
\begin{remark}
The result of this lemma may be optimal, deterministically speaking. But, we expect equivalence classes of $G_r$ of diameter close to $r$ to be very rare. We probably lose a large factor in this step. However, this intuition seems difficult to quantify.
\end{remark}
\begin{proof}[Proof of Lemma \ref{lem:BoundDist}]
In this proof, we use the fact that $f$ satisfies $\Delta f+f=0$ almost surely. Without loss of generality, we may suppose that $f(x)> 0$. We claim that $f(z)<0$ for all $z$ such that $|x-z|=r$. We already know that $f(z)$ has constant sign for all $z$ such that $|x-z|=r$. Consider the function
$$g(z) := \frac{1}{2\pi} \int_0^{2\pi} f(R_\theta z) \mathrm{d}\theta,$$
where $R_\theta$ is the rotation of centre $x$ and of angle $\theta$. The function $g$ satisfies $(\Delta+1)g=0$ (because $f(R_\theta\cdot)$ does for each $\theta\in[0,2\pi]$), and it is radially symmetric around $x$. Therefore, we must have
$$g(z)= \lambda J_0 (|z-x|)$$
for some $\lambda \neq 0$. Since $f(x)>0$, we must have $\lambda>0$, and hence, $f(z)<0$ for all $z$ such that $|x-z|=r$. Therefore, the connected component of $\R^2\backslash f^{-1}(\{0\})$ containing $x$ is included in $B(x,r)$ and the Lemma follows.
\end{proof}
We now finish off the proof of Theorem~\ref{t.main} by estimating the expected size of $G_r$ using Lemma~\ref{l.good}. Let $0<r_1<r_2$ be the first two zeros of the Bessel function $J_0$ and fix $r\in]r_1,r_2[$. Then, by Lemma \ref{lem:BoundDist}, for each $x\in G_r$, the equivalence class of $x$ has diameter at most $r$. By the isodiametric inequality (see paragraph 10 of \cite{bonneson_fenchel}) its area is no greater than $\frac{\pi}{4}r^2$. Also, two different equivalence classes are included in different connected components of $\R^2\backslash f^{-1}(\{0\})$. Finally, if $R>r$ and $x\in B(0,R-r)$, then the connected component of $\R^2\setminus f^{-1}\left(\{0\}\right)$ containing $x$ is included in $B(0,R)$. Thus, for each $R>r$,
\[
\vol\left(G_r\cap B(0,R-r)\right)\leq\sum_c\vol(c)\leq\frac{\pi}{4}r^2\times N(R,f)
\]
where the sum runs over the equivalence classes of $G_r$ intersecting $B(0,R-r)$. Taking expectations, by stationarity, we get
\[
\vol\left(B(R-r)\right)\prob\left[0\in G_r\right]\leq\frac{\pi}{4}r^2\times\E\left[N(R,f)\right]\, .
\]
Dividing by $\vol\left(B(0,R)\right)=\pi R^2$ and letting $R\rightarrow+\infty$, we get $\prob\left[0\in G_r\right]\leq\frac{1}{16}r^2\nu_{BS}$ which yields the following lower bound for the Bogomolny-Schmit constant:
\[
\nu_{BS}\geq \frac{16 \prob[0\in G_r]}{r^2}\, .
\]
By Lemma~\ref{l.good}, we have, for each $T>0$ and each $r>r_0$,
\[
\nu_{BS}\geq\frac{32}{r^2}\left[\Psi(T)-\frac{r}{\sqrt{2}}\Psi\left(\frac{T}{\sqrt{1-J_0(r)^2}}\right)\right]\, .
\]
Taking $r=3.8$ (the first minimum of $J_0$) and $T=3.35$ (the smallest $T$ for which in (\ref{eq:LowerProp}), the function whose expectation we compute is always positive), we get
\[
\frac{32}{r^2}\geq  2.216\, ;\ \frac{T}{\sqrt{1-J_0(r)^2}}\geq 3.659\, ;\ \frac{r}{\sqrt{2}}\leq 2.69\, .
\]
We therefore obtain the announced lower bound
\[
\nu_{BS}\geq 1.39\times 10^{-4}\, .
\]
\bibliographystyle{alpha}
\bibliography{ref}

\end{document}